\documentclass[11pt,letterpaper]{article}

\pdfoutput=1

\usepackage[margin=1in]{geometry}

\usepackage{amsmath, amsthm, amssymb, amstext, comment, graphicx}
\usepackage{amssymb}

\newtheorem{proposition}{Proposition}
\newtheorem{corollary}{Corollary}
\newtheorem{theorem}{Theorem}
\newtheorem{lemma}{Lemma}

\theoremstyle{definition}
\newtheorem{definition}{Definition}

\title{The Demand Query Model for Bipartite Matching}

\author{Noam Nisan\footnote{School of Computer Science and Engineering, The Hebrew University of Jerusalem.
Supported by the European Research Council (ERC) under the European Union’s Horizon 2020 research and innovation programme grant agreement No 740282.}}
\begin{document}
\maketitle

%\section{Overview}
\begin{abstract}
We introduce a ``concrete complexity'' model for studying algorithms for matching in bipartite graphs.  The model
is based on the ``demand query'' model used for combinatorial auctions.  Most (but not all) known algorithms for bipartite matching seem to be translatable into this model including exact, approximate, sequential, parallel, and online ones.  

A perfect matching in a bipartite graph can be found in this model with $O(n^{3/2})$ demand queries
(in a bipartite graph with $n$ vertices on each side)
and our main open problem is to either improve the upper bound or prove a lower bound.  An
improved upper bound could yield ``normal'' algorithms whose running time is better
than the fastest ones known, while a lower bound would rule out a faster algorithm
for bipartite matching from within a large class of algorithms.

Our main result is a lower bound for finding an approximately maximum size matching in parallel: A {\em deterministic}
algorithm that runs in $n^{o(1)}$ rounds, where each round can make
at most $n^{1.99}$ demand queries cannot find a matching whose size is within $n^{o(1)}$ factor
of the maximum.  This is in contrast to {\em randomized} algorithms
that can find a matching whose size is $99\%$ of the maximum in $O(\log n)$ rounds, each making $n$ demand
queries.
\end{abstract}

\section{Introduction}

In the (unweighted, bipartite) matching problem, a bipartite graph with $n$ left vertices and $n$ right
vertices is given, and the problem is to find a maximum-size matching:
a set of edges of the graph of largest cardinality such that no two of which share a vertex.  This problem
has numerous applications and has been studied extensively in a variety of computational models including
sequential, parallel, online, and approximate.  Many extensions of the problem (e.g. to 
weighted or to non-bipartite graphs) have been studied as well.

The fastest known deterministic algorithm for the problem \cite{HK73} is close to half a century old
and runs in time $O(n^{5/2})$.  It is a long standing open problem whether this running time may 
be improved, but a
faster {\em randomized} algorithm whose running time is $n^\omega$ (where $\omega=2.3...$ is the matrix multiplication exponent) was
given in \cite{MS04}.  Another long standing open problem is whether there 
exists a deterministic NC algorithm
(a parallel algorithm running in poly-logarithmic time using a polynomial number of processors)
for finding a maximum matching. Randomized such algorithms were given in \cite{KUW85,MVV87}.

While an enormous amount of algorithmic work was done on many variants of the matching problem,
there is currently no hope for proving hardness results: we simply lack any tools that 
can prove lower bounds for general algorithms.  One approach for progress is to define a {\em concrete model of
computation} -- one that is strong enough to capture many of the known algorithms -- and try to 
understand the complexity in that model.  This is our approach here. 

\subsection{The Demand Query Model}

Our computational model is inspired by an economic point of view.  We consider the right vertices of our bipartite graph 
to be items and the left vertices to be {\em agents}, each who is interested in acquiring exactly a single item from
some subset of items -- those that are adjacent to it in the graph.  In a general economic setting, a ``demand query'' 
(see e.g., \cite{BN07}) asks an economic agent the following type of question: suppose that 
each item $j$, from some set of items, could be 
purchased for price $p_j$ -- which set of items would you demand to purchase?  In our setting, we assume that
our agents are ``unit demand'', i.e. desire a single item, and as our graphs are unweighted our agent will
simply choose the least expensive item.  Thus in our setting we consider the following type of query on a biprtite graph.

\begin{definition}
A {\em Demand Query} accepts a left vertex $v$ and an order on the right vertices $(u_1...u_{n})$
and returns the index of the first vertex $i$ in the order such that there is an edge $(v,u_i)$ in 
the graph, or $0$ to denote that none exists.\footnote{For compatibility with the notion of demand queries, 
we only allow $v$ to be a left vertex.  In the graph context it may be natural to consider also variants of our model
where $v$ can be any vertex in the graph.  It turns out that our main lower bound applies to this stronger model as well.}
\end{definition}  

Desiring a concrete model, the only cost of an algorithm that we will count is the number of demand queries required. 
(It will turn out that none of our upper bounds have any other significant algorithmic costs.) 
In particular every problem in this model has an upper bound of $O(n^2)$ as a single demand query can certainly tell
us whether an edge $(v,u)$ exists in the graph.  

We first observe that many of the classic algorithms for matching can be ``implemented'' in this model, including
augmenting paths methods and primal-dual auction-like algorithms \cite{B88,DGS86}.  These can give
us the following basic upper bound:

\begin{theorem}
A maximum size matching in a bipartite graph with $n$ left and $n$ right vertices can be found using $O(n^{3/2})$
demand queries.
\end{theorem}

\subsection{From the Demand Query Model to General Algorithms}

When converting an algorithm in the demand query model into a ``real algorithm'' that does not have a
demand query as a primitive, but rather must implement it over a normal computation model, each demand query can be 
trivially implemented in $O(n)$ time by going over the edges of the queried vertex and finding the one with
minimum rank in the list.  Using this $O(n^{3/2})$ bound with this simulation
matches the currently best $O(n^{5/2})$ algorithm for bipartite matching \cite{HK73}.\footnote{One may certainly 
attempt to obtain faster simulations of the demand query model by normal algorithms using, say, clever data structures,
but we have not been able to do so in general.}
While we generally focus on the case of dense graphs, we note that for sparse graphs with $m << n^2$ edges, one may also obtain the improved known upper bound of $O(mn^{1/2})$, as a demand
query to vertex $v$ can in fact be simulated in $O(d_v)$ time, where $d_v$ is the degree of $v$ in the graph.\footnote{
Getting the $O(mn^{1/2})$ total running time
is slightly non-trivial since this requires that high-degree vertices are not queried too 
often by the simulated algorithm.  In case
that they are, some data structure work will be needed to even things out in an amortized sense \cite{N09}.}
We also find that the complexity of finding an {\em approximately} maximal matching,
one  whose size is at least $(1-\epsilon)$ of the maximal in the demand query model is only $O(n/\epsilon)$. 
Again, with the direct
simulation of demand queries by ``normal'' algorithms, this implies the (known)
$O(n^2)$ algorithm (linear in the input size) for finding a matching of size 
at least $99\%$ of maximal.

While most algorithms for bipartite matching seem to be translatable into this model, a known algorithmic technique that does {\em not} seem to fall under this model is an algebraic one relying on matrix multiplication.  A simple example is solving the decision problem of whether a perfect matching exists using randomization to test whether the symbolic determinant of the 
graph is identically zero \cite{L79} which can be done in time $O(n^\omega)$ (where $\omega=2.3...$ is
the matrix multiplication exponent).  A randomized algorithm with the same running time 
that actually finds a maximum matching 
is given in \cite{MS04}.\footnote{The distinction between the decision problem and the search problem is not
significant for anything in this paper though: all upper bounds actually find a matching, and the lower bound
applies even for the decision variant.}  These algorithms are all randomized and it is not clear
whether deterministic algorithms can match this performance.  
Another known technique that dose not seem to fall in this model is using interior point methods,
which in \cite{M13} was used for giving an $O(m^{10/7})$ algorithm, where $m$ is the total 
number of edges in the graph, which beats 
the $O(n^{5/2})$ bound for very sparse graphs.

\subsection{Lower Bounds?}

A lower bound of 
$\Omega(n)$ for bipartite matching in the demand query model is trivial, even for the approximate version of the problem\footnote{E.g. for distinguishing between a maximum matching of any positive size and zero size, i.e. whether the graph
is empty.}.  Our {\em main open problem} is to close the gap between this trivial lower bound and the
$O(n^{3/2})$ upper bound: either improve the upper bound, which
would likely imply faster normal algorithms\footnote{In principle,
as the demand query model formally allows ``non-uniform'' algorithms, this simulation by normal algorithms is not a formal theorem, but we do not know of any example where this non-uniformity is used.}, or prove a lower bound, which would rule out faster algorithms
from a rather wide class of algorithms.  

\vspace{0.1in}
\noindent
{\bf Open Problem 1:}
Does there exists a faster algorithm for finding a maximal matching in the demand query model? (I.e. one 
that uses $o(n^{3/2})$ queries,  
or perhaps even $O(n)$ queries?)
\vspace{0.1in}

A hint that an answer may be nontrivial is that the nondeterministic 
(certificate) complexity of maximum matching
in this model is only $\Theta(n)$: a maximum matching can be given by listing the $n$ edges in it,
where each edge can be certified by a single demand query, and, by Hall's theorem, the maximality of said matching 
of size $k$ can be 
certified by exhibiting a set vertices (left and right) of total size $2n-k$ with no edges between them, which 
can be certified using one demand query for each left vertex in the set. 

\subsection{Related Models}

The demand query model lies between two well studied models: Boolean decision trees \cite{BD02} and 
communication complexity \cite{KN96}.
In the Boolean decision tree model the allowed queries are only to individual edges, i.e. may ask whether an edge 
$(v,u)$ exists in the graph.  In the communication complexity model the left vertices are treated as agents 
that each ``knows''
the set of edges connected to it and considers the amount of communication needed between these agents.  This
would be equivalent to a decision tree model that allows arbitrary queries about the set of egdes adjacent to
any single left vertex.  

In the decision tree model, it is not difficult to see that the complexity of even determining whether
a perfect matching exists
is $\Theta(n^2)$, even non-deterministically (for certifying none existence).  
In the communication complexity model what is known is identical to what was described above in the 
demand query model, 
and one may consider
the demand query model as an intermediate step towards understanding the complexity in the communication
complexity model.

A model that turns out to be equivalent, up to log factors, to the demand query model is a model that allows ``OR''
queries on the edges of a single left vertex.  I.e. an ``OR'' query specifies a left vertex $v$ and a set $S$
of right vertices and asks whether there exists an edge between $v$ and some vertex in $S$. This binary query 
is clearly weaker than the demand query, but it is not difficult to
use binary search to simulate demand queries using $O(\log n)$ ``OR'' queries.  It will be convenient to prove our
lower bound in the binary ``OR'' query model which them implies the lower bound in the demand query model.

A slightly stronger model (which is incomparable to the communication complexity model, but generally seems 
rather weaker) would allow asking about the ``OR'' of an arbitrary set of edges, not just those connected to
a single vertex.  I.e. such a query can present an arbitrary subset of $n^2$ possible edges and ask whether at least
one of them exists in the graph.  It is known \cite{LM19} that such a model is equivalent, up to log factors, to the
following complexity measure in Boolean decision trees: the maximal number of 1-answers on the path to any leaf of
the tree.  

An even stronger model would allow asking for an $OR$ of a subset of edges and their negations (e.g. ``is 
either $(u_1,v_1)$ an edge in the graph OR $(u_2,v_2)$ not an edge?'').  It can be shown
that this model captures, up to a log factor, the logarithm of minimum Boolean decision tree {\em size}.  
We do not know any improved bounds in these models and  
proving an improved upper bound in one of them is an open problem as well.

\subsection{Parallel Algorithms in the Demand Query Model}

Up to this point we have introduced our general model, its motivation, and the main open problem.  We now
move onto the issue of the parallel complexity in this model, which is our main technical result.

A parallel algorithm in the demand query model proceeds in rounds, where at each round multiple demand 
queries can be made.
The point is that the identity of all the queries within a single round must all be determined only based on the
outcomes of queries from the previous rounds.  The two basic parameters in a parallel algorithm are thus the
number of rounds and the (maximum allowed) number of queries per round.  Any sequential algorithm that
makes a total of $q$ queries can be viewed as a parallel algorithm with $q$ rounds, each of 1 query.  
Every
parallel algorithm with $r$ rounds of $q$ queries each, can be trivially converted to a
sequential algorithm with $r \cdot q$
queries, but the converse is not necessarily true.  

Any problem (on graphs with with $n^2$ possible edges) can be trivially solved with a single round 
of $n^2$ queries.  The question that we ask is whether a maximum matching in a bipartite graph 
(or at least an approximation thereof) can be found
with ``few'' rounds or ``few'' queries each, say 
is $n^{o(1)}$ rounds of $n^{1+o(1)}$ demand queries per
round.  (Recall that the trivial lower bound on the total number of queries
is $\Omega(n)$.)  In \cite{DNO14} a {\em randomized} algorithm was presented that operates within fully within this model
and within $O(\log n / \epsilon^2)$ rounds, each of $n$ demand queries (a single query per left vertex)
outputs, with high probability, a matching whose size is at least $(1-\epsilon)$ fraction of the maximum
matching.  

\vspace{0.1in}
\noindent
{\bf Open Problem 2:}
Does there exists a randomized parallel algorithm that uses ``few'' 
(even $O(\log n)$)
rounds each with ``few'' (even $O(n)$) demand queries and outputs, with high probability, 
a perfect matching if one exists.
\vspace{0.1in}

\subsection{Main Result: A Lower Bound for Deterministic Parallel Algorithms}

Our main result is a lower bound showing that randomization is crucial here: no deterministic algorithm
that uses ``few'' rounds of ``few'' queries can get a decent approximation to the maximum matching size.
In fact our lower bounds show that the algorithms cannot even distinguish between graphs with a perfect
matching and those with a small maximum matching size.
We provide a general tradeoff between the number of rounds, the number of queries per round, and the level of
approximation possible which in particular implies:

\vspace{0.1in}
\noindent
{\bf Instance of main theorem with a specific choice of parameters:}
Deterministic algorithms with $n^{1/7}$-rounds each using $n^{8/7}$ demand queries cannot approximate 
the maximum
matching size within a factor of $n^{1/7}$.
\vspace{0.1in}

Beyond tightening our tradeoffs, a natural challenge is to extend our lower bound to the stronger communication
complexity model.   Our lower bound uses a direct adversary techniqe which does not seem to extend to the communication
model.

\vspace{0.1in}
\noindent
{\bf Open Problem 3:}
Does there exists a deterministic {\em communication protocol} between the $n$ agents (left vertices) 
that uses ``few'' (even $O(\log n)$)
rounds of communication where in each round each player sends ``few'' (even $O(1)$) bits
of communication and outputs, with high probability, a perfect matching if one exists.
\vspace{0.1in}

Improving upon \cite{DNO14, ANRW15}, in \cite{BO17} 
a nearly logarithmic lower bound on the number of rounds was proven if each round
allows, say, $polylog(n)$ bits of communication per player.  This bound is exponentially
lower than the query bound that we obtain, and unlike our bound, applies also to randomized 
protocols.  A gap between deterministic and randomized {\em simultaneous} communication protocols was
exhibited in \cite{DNO14}.

\section{Some algorithms}

\subsection{Warmup: The Greedy Online Algorithm}

We start by looking at the simplest greedy algorithm that also has the advantage of being ``online'' \cite{KVV90}: we can
have the left vertices come one after another in an online fashion, and as each vertex arrives we immediately 
match him to an arbitrary yet unmatched right vertex.  This can be done using a single demand query per 
vertex: the query will place all the unmatched vertices (in an arbitrary order) before all the matched ones.
This is known to produce a matching whose size is at least 1/2 of the maximum size.  A randomized variant 
fixes a random order on the right vertices and allocates the first 
unmatched vertex rather in this order than an arbitrary one (still a single demand query per vertex, 
but now with the fixed random
order on the yet unmatched vertices), and it is known that this produces, in expectation, a $(1-1/e)$-approximation
to the maximum matching \cite{KVV90}. 

\begin{proposition}
There is a deterministic online algorithm for maximum matching
that uses 1-demand query per vertex and provides a 1/2-approximately maximum matching.  
There is a randomized online algorithm for maximum matching
that uses 1-demand query per vertex and provides a (1-1/e)-approximately maximum matching.
\end{proposition}

\subsection{The Ascending Auction Algorithm}

We now present the ``ascending auction'' algorithm of 
\cite{DGS86,B88} which 
is naturally described in the demand-query model, and may be viewed as motivating this model.\footnote{This algorithm is quite naturally interpreted as an
ascending auction that can be used also in more general scenarios \cite{KC82} and in
our scenario turns out to also has nice incentive properties \cite{DGS86}.}  It will be beneficial to present the
algorithm as an approximation algorithm that is parametrized by the desired approximation ratio.

\begin{theorem}
For every parameter $0 < \epsilon < 1$, there exists an algorithm that uses $O(n/\epsilon)$ demand queries and
produces a matching whose size is at least $(1-\epsilon)$ fraction of the maximum size matching.
\end{theorem}

Here is the algorithm essentially due to \cite{DGS86}:

\begin{enumerate}
\item For each left vertex $u$ initialize ``prices'' $p_u=0$.
\item Initialize the matching $M=\emptyset$.
\item Initialize the set of ``TBD'' vertices $A$ to be the set of all right vertices.  
\item While $A$ is not empty:
      \begin{enumerate}
			\item Pick an arbitrary vertex $v \in A$ and remove it from $A$.
			\item Ask a demand query on $v$ with the order on $u$'s induced by increasing values of $p_u$ (lowest price first, and ties broken arbitrarily), and let $u$ be the answer to that query.
			\item If $p_u < 1$: 
			    \begin{enumerate}
					\item If for some vertex $v'$ we have that $(v',u)\in M$, insert $v'$ into $A$ and remove $(v',u)$ from $M$.
					\item Insert $(v,u)$ into $M$.
					\item Increase $p_u$ by $\epsilon$
					\end{enumerate}		
			\end{enumerate}
\item Output $M$
\end{enumerate}

For completeness, the correctness of this algorithm is sketched below.  (For ease of exposition let us assume that 
$\epsilon = 1/t$ for some $t$.)
 
\begin{proof}
The invariant that we keep is that for every left vertex $v \not\in A$ we have one of the following
two cases: either (1) for all existing edges
$(v,u)$ we have that $p_u = 1$ or (2) for some edge $(v,u)$ we have that there exists a unique
right vertex $u$ such that $(v,u)\in M$
and for all other edges $(v,w)$ it holds that $p_w \ge p_u-\epsilon$.

For an arbitrary matching $N$ let us define $util(v,N)=1-p_u$ for $(v,u)\in N$ and $util(v,N)=0$ is $v$
is unmatched in $N$.  It is easy to see that our invariant ensures that for any matching $N$ we have
that $util(v,N) \le util(v,M) + \epsilon$ (where $M$ was the matching produced by the algorithm), and in fact
for $v$ that is unmatched in $N$ we don't even loose $\epsilon$: $util(v,N) = 0 \le util(v,M)$.
Let us now sum $util(v,N)$ up over all left vertices $v$: we get $|N|$ times 1 minus $\sum p_u$ over all $u$
that are matched in $N$ which is bounded from above by $\sum p_u$ over {\em all} right vertices $u$.  For the case of the matching $M$ found by the algorithm any unmatched right vertex $u$ still has $p_u=0$ so we get exactly 
$\sum p_u$ over all left vertices $u$.  We thus have 
$|N|-\sum_u p_u \le \sum_v util(v,N) \le \sum_v util(v,M) + |N|\epsilon = |M| - \sum_u p_u + |N|\epsilon$.
It follows that $|M| \ge (1-\epsilon)|N|$.
\end{proof}

In terms of the running time, every iteration of the main loop makes a single demand query and either increases the
price $p_u$ of some right vertex by $\epsilon$ or removes a left vertex (forever) from $A$.  There can clearly
be at most $n$ iterations that remove a vertex and, since the price of any
right vertex never increases above 1, there can be at most $n/\epsilon$ price increase iterations, for
a total running time of $O(n/\epsilon)$.

In terms of obtaining the perfectly maximum matching, one could take $\epsilon = 1/(n+1)$ which 
due to the integrality of the matching size would imply that
the algorithm produces the maximum matching.  This makes sense as a way of obtaining a ``usual'' $O(n^3)$-time
algorithm for maximum matching (when each demand query is executed in $O(n)$ time), but in our model this would
require $O(n^2)$ demand queries which in our model is trivial.

\subsection{Augmenting Paths}

We now present a variant of
the directed connectivity problem, which is used as a step (``augmenting path'') in many algorithms for bipartite matching: We are
given some mapping from the right vertices to the left vertices $\pi:R \rightarrow (L \cup \{\Lambda\})$ and are given
a subset $S$ of left vertices.  
Our goal is to find a directed path $(v_1, u_1, v_2, u_2, ... , v_k, u_k)$
such that $v_1 \in S$, $\pi(u_k)=\Lambda$, and 
for every $1 \le i < k$ we have that the edge $(v_i, u_i)$ is in the graph and 
$\pi(u_i) = v_{i+1}$, or say that none exists.

\begin{lemma}
For every $S$ and $\pi$, this problem can be solved with $O(n)$ demand queries.
\end{lemma}

\begin{proof}
We will implement breadth first search in the demand query model:

\begin{enumerate}
\item Initialize a FIFO Que $Q$ with all left vertices in $S$, and initialize a set of ``discarded'' 
left vertices $D=\emptyset$.
\item Repeat until $Q = \emptyset$:
\begin{enumerate}
\item Let $v$ be the first vertex in $Q$.
\item Pick an arbitrary order of the right vertices such that all vertices $u$ with $\pi(u)=\Lambda$
appear before all others and then appear all vertices with $\pi(u) \not\in (D \cup Q)$ and last appear those with
$\pi(u) \in (D \cup Q)$, and ask a demand query
from $v$ in this order.  
\item If $\pi(u)=\Lambda$ Then output the path leading to $v$ and then $u$ and halt.
\item If $\pi(u) \not\in (D \cup Q)$ Then enque $\pi(u)$ into $Q$, attaching to it the path first leading to 
$v$, then $u$ and then $\pi(u)$.
\item If $\pi(u) \in (D \cup Q)$ Then remove $v$ from $Q$ and insert $v$ into $D$.
\end{enumerate}
\item If a path was not found and the loop terminated due to $Q$ being empty, there is no such path.
\end{enumerate}

In each round we make a single demand query and either insert a new vertex into $Q$ (which we can do
at most $n$ times) or remove an element from $Q$ (which again we can do at most $n$ times). 
\end{proof}

Note: We could have alternatively simulated depth first search.  
It is not difficult to prove a matching lower bound for this problem.

This version of the connectivity problem is exactly what is needed for an augmenting path step used in many 
matching algorithms: 
A partial matching defines our mapping $\pi$ by having $\pi(u)$ be defined as the vertex that is matched to it 
in the partial matching (and $\Lambda$ is $u$ is not matched) and defines the set $S$ of unmatched left vertices. 
As it is known that a matching in a bipartite graph has maximum size if and only it admits no augmenting
path, this gives a $O(n)$ algorithm for testing whether a given matching has maximum size.

\subsection{An $O(n^{3/2})$-query algorithm}

We now have all the ingredients for describing our best algorithm for maximum matching in this model.

\begin{theorem}
There exists an algorithm for maximum matching in a bipartite graph that uses $O(n^{3/2})$ demand queries.
\end{theorem}

This is obtained by first running the ascending auction algorithm with $\epsilon=1/\sqrt{n}$.  This requires
$O(n^{3/2})$ demand queries and produces a matching whose size is at least $(1-1/\sqrt{n})$ times
the maximum size.  We then run a sequence of augmenting path steps, each of which requires $O(n)$ additional
demand queries.  Notice that every augmenting path step increases the matching size by 1, and we started 
with a matching that can be smaller than the maximum matching by at most an additive $O(\sqrt{n})$,
at most $\sqrt{n}$ augmenting path steps are needed before the maximum matching is obtained, for a total
of at most $O(n^{3/2})$ demand queries.

\subsection{A Randomized Parallel Algorithm}

In the parallel verion of our model
we proceed by {\em rounds}, where at each round a set of 
queries is asked, a set that may be determined by the answers to the queries 
from previous rounds.  An $r$-round $q$-query-per-round protocol is one with at most $r$ rounds, where at each
round at most $q$ demand queries are made. 

\begin{theorem}(Dobzinski-Nisan-Oren) for any $\delta>0$ there exists a {\em randomized}
$O(\log n/\delta^2)$-round protocol where at each round there is a single demand query for every left vertex
(for a total of $n$ demand queries per round)
which returns a matching of size at least $(1-\delta)$-fraction
of the optimal matching.
\end{theorem}

The algorithm is a randomized parallel variant of the ascending auction algorithm:

\begin{enumerate}
\item For each left vertex $u$ initialize ``prices'' $p_u=0$.
\item Initialize the matching $M=\emptyset$, and the discarded vertices $D=\emptyset$.
\item Repeat $O(\log n/\delta^2)$ times:
\begin{enumerate}
\item For each vertex $v \not\in D$ that is currently unmatched in $M$, in parallel:
      \begin{enumerate}
			\item Ask a demand query on $v$ with the order on $u$'s induced by increasing values of $p_u$, 
			with ties broken {\em randomly}, and let $u_v$ be the answer to that query.
			\item If $p_{u_v} > 1$: insert $v$ into $D$.		    		
			\end{enumerate}
\item For each $u$ such for some $v$ we have that $u=u_v$, increase the price $p_u$ by $\delta$ and
pick an arbitrary $v$ such that $u=u_v$ and insert
$(v,u)$ into the matching $M$, removing any previous edge matched to $u$, if any.
\end{enumerate}
\end{enumerate}

\section{The Lower Bound}

It will be more convinient to prove our lower bound with a slightly weaker query, the OR query, which turns out
to be essentially equivalent to a demand query, (up to log factors).

\begin{definition}
An {\em OR Query} accepts a left vertex $v$ and a subset $S$ of the right vertices and 
returns whether there exists a vertex
$u \in S$ such that $(v,u)$ is an edge in the graph.
\end{definition}

While it is clear that a demand query is stronger than an OR query, it turns out that the gap between them is not large:

\begin{lemma}
A demad query can be simulated by $\log_2 (n+1)$ OR queries.
\end{lemma}

\begin{proof}
We simulate each demand query $(v,(u_1...u_n))$ by a binary search that uses OR queries: we start by asking whether
there is an edge between $v$ and $\{u_1...u_{n/2}\}$, according to the answer we then either focus on the first half
of the variables (asking whether there is an edge to $\{u_1...u_{n/4}\}$) or the second half (asking whether there
is an edge to $\{u_1...u_{3n/4}\}$), etc.
\end{proof}

The overhead in the simulation is clearly optimal since a demand query has $n$ possible answers.  We will proceed
by providing our lower bounds in the OR query model, which
as we have just shown implies the same lower bounds -- with a log factor loss -- in the demand query model.

\begin{theorem}
Every deterministic $r$-round $q$-query (per round) algorithm that uses OR queries cannot distinguish between grpahs with a perfect matching (of size $n$) and those with a matching of size at 
most $\alpha n$ for
$\alpha = r \sqrt{q r \log n} / n$. 
\end{theorem}

\begin{corollary}
Deterministic algorithms with polylog rounds of polylog demand queries per-player can 
only find matching of size $\tilde{O}(\sqrt{n})$ in a graph that has a perfect matching.
\end{corollary}

\begin{corollary}
Deterministic $r$-round algorithms require at least $q = n^2/(r^3 \log^4 n)$  demand queries per round in order to find a
a maximum matching in a bipartite graph or even find a matching whose size is a constant fraction of the maximum
size matching.
\end{corollary}

\begin{corollary}
Deterministic $n^{1/7}$-round $n^{8/7}$ queries-per-round algorithms cannot approximate the maximum
matching within a factor of $n^{1/7}$.
\end{corollary}

\begin{proof}
We will describe an adversary algorithm for answering the rounds of queries.  
At every round, our adversary will 
decide (based on the queries in this round) on the existence or lack of existence of some 
subset of the edges
of the graph, in a way that the answers to all queries in this round are determined by these choices.
The edges that were decided to exist in this round will be called "YES" egdes of the round, and
those that were decided not to exist will be called "NO" edges of the round.

For any query to a set $S$ of edges that contains some edge that was already decided to exist
(a "YES" edge from some previous round), the answer to this query is already determined so
the adversary can
simply ignore
this query in this round (and thus the discussion below can assume that such sets are never queried).  
Similarly any query to a set $S$ that
includes some subset of "NO" edges from previous rounds may be considered
by the adversary as though it is a query only
to the subset of edges that were previously not answered.  The rest of the discussion can thus assume 
without loss of generality that all edges queried in all queries in this round are ``new'' ones
for which the adversary has not committed to an answer yet.

Specifically, 
for every $OR$ query on a set $S$ of edges made in this round the adversary will either 
fix all the edges in $S$ to be "NO" edges (which fixes a negative answer to this query), or will fix
(at least) one edge in $S$  to be "YES" (which fixes a positive anser to this query).  The adversary's
goal is to make sure that after fixing all these "YES" and "NO" edges in all $r$ rounds, the underlying
graph can still either have a perfect matching or not have any matching of size greater than $\alpha n$.
The adversary will maintain this property by maintaining the following constraints on the 
"YES" and "NO" edges of {\em each round}:

\begin{enumerate}
\item For every vertex in the graph, strictly less than $n/(2r)$ edges adjacent to it are "NO" edges of the round.
\item The set of vertices that are adjacent to any "YES" edge of the round is of size at most $O(\alpha n /r)$.
\end{enumerate}

Given these conditions, the algorithm cannot distinguish between the following two extreme cases: (a) all
edges that were not decided in any round in fact do not exist, in which case the maximum possible matching is of size
$O(\alpha n)$ since every edge in it must be adjacent to a "YES" edge of some round, and there are at most
$O(\alpha n / r)$ such vertices in every round (b) all edges that 
were not decided yet do exist in which case a perfect matching exists (which follows from the fact that the degree
of every vertex is strictly more that $n/2$ since in each round strictly less than $n/(2r)$ edges adjacent to
any vertex were set to be ``NO'' edges.)

Here is how the adversary makes decisions for a given round that makes $q$ 
$OR$ queries on sets $S_1 ... S_q$ (each of which contains only edges that were not answered to be ``YES'' or ``NO''
in any previous round).
The adversary will handle separately the $OR$ queries on sets that contain at least 
$\theta = n \sqrt{\log n}/\sqrt{2qr}$ edges (``big queries'') 
and those of size at most $\theta$ edges (``small queries'').  

\noindent
{\bf Big queries:} the adversary will pick a set of size $O(n\log n / \theta)$ of vertices 
with the property that it intersects each of the ``Big'' queries $S_i$ and answer that all edges 
connected to them (those which were not already decided) exist, i.e are ``YES'' edges for this round.  
This suffices for answering all the big queries with a positive answer.
The existence of such a set 
is given by a random construction: choose every vertex at random to be in this set with probability 
$O(\log n / \theta)$:
for a fixed big query the probability that at non of its edges 
connects to this chosen set is at most $(1-\log n /\theta)^\theta < 1/n^2$ and since there are certainly less
than $n^2$ queries, with high probability the chosen set intersects all big queries simultaneously. 

\noindent
{\bf Small queries:} At this point the adversary has already ensured a positive answer to all
big queries, and perhaps also to some small queries, so let us focus on the small queries 
for which the answer is not yet fixed.  There are at most $q \theta$ edges in all of the these small 
queries combined.  Since each 
edge contains two vertices, there can be at most $(2 \cdot q \theta) / (n/(2r)) = 2 \sqrt{2qr\log n}$ 
vertices in the graph that are each adjacent 
to more than $n/(2r)$ of these edges, which will be called ``heavy'' vertices.  
The adversary will answer ``YES'' to all edges that are 
adjacent to one of these 
heavy vertices
and ``NO'' to all other edges. So we only need to show that there are at most $O(\alpha n /r)$ such heavy
vertices, which is true since $\alpha n / r = \sqrt{q r \log n}$ as needed.
\end{proof}

\bibliographystyle{plain}
\bibliography{bib}

\begin{thebibliography}{10}

\bibitem{ANRW15}
Noga Alon, Noam Nisan, Ran Raz, and Omri Weinstein.
\newblock Welfare maximization with limited interaction.
\newblock In {\em 2015 IEEE 56th Annual Symposium on Foundations of Computer
  Science}, pages 1499--1512. IEEE, 2015.

\bibitem{B88}
Dimitri~P Bertsekas.
\newblock The auction algorithm: A distributed relaxation method for the
  assignment problem.
\newblock {\em Annals of operations research}, 14(1):105--123, 1988.

\bibitem{BN07}
L.~Blumrosen and N.~Nisan.
\newblock Combinatorial auctions (a survey).
\newblock In N.~Nisan, T.~Roughgarden, E.~Tardos, and V.~Vazirani, editors,
  {\em Algorithmic Game Theory}. Cambridge University Press, 2007.

\bibitem{BO17}
Mark Braverman and Rotem Oshman.
\newblock A rounds vs. communication tradeoff for multi-party set disjointness.
\newblock In {\em 2017 IEEE 58th Annual Symposium on Foundations of Computer
  Science}, pages 144--155. IEEE, 2017.

\bibitem{BD02}
Harry Buhrman and Ronald De~Wolf.
\newblock Complexity measures and decision tree complexity: a survey.
\newblock {\em Theoretical Computer Science}, 288(1):21--43, 2002.

\bibitem{DGS86}
Gabrielle Demange, David Gale, and Marilda Sotomayor.
\newblock Multi-item auctions.
\newblock {\em The Journal of Political Economy}, pages 863--872, 1986.

\bibitem{DNO14}
Shahar Dobzinski, Noam Nisan, and Sigal Oren.
\newblock Economic efficiency requires interaction.
\newblock In {\em Symposium on Theory of Computing, {STOC} 2014, New York, NY,
  USA, May 31 - June 03, 2014}, pages 233--242, 2014.

\bibitem{HK73}
John~E Hopcroft and Richard~M Karp.
\newblock An n\^{}5/2 algorithm for maximum matchings in bipartite graphs.
\newblock {\em SIAM Journal on computing}, 2(4):225--231, 1973.

\bibitem{KUW85}
Richard~M Karp, Eli Upfal, and Avi Wigderson.
\newblock Constructing a perfect matching is in random nc.
\newblock In {\em Proceedings of the seventeenth annual ACM symposium on Theory
  of computing}, pages 22--32. ACM, 1985.

\bibitem{KVV90}
Richard~M Karp, Umesh~V Vazirani, and Vijay~V Vazirani.
\newblock An optimal algorithm for on-line bipartite matching.
\newblock In {\em Proceedings of the twenty-second annual ACM symposium on
  Theory of computing}, pages 352--358. ACM, 1990.

\bibitem{KC82}
Alexander~S Kelso~Jr and Vincent~P Crawford.
\newblock Job matching, coalition formation, and gross substitutes.
\newblock {\em Econometrica: Journal of the Econometric Society}, pages
  1483--1504, 1982.

\bibitem{KN96}
Eyal Kushilevitz and Noam Nisan.
\newblock {\em Communication complexity}.
\newblock Cambridge university press, 1996.

\bibitem{LM19}
Bruno Loff and Sagnik Mukhopadhyay.
\newblock Lifting theorems for equality.
\newblock In {\em 36th International Symposium on Theoretical Aspects of
  Computer Science, {STACS} 2019, March 13-16, 2019, Berlin, Germany}, pages
  50:1--50:19, 2019.

\bibitem{L79}
L{\'a}szl{\'o} Lov{\'a}sz.
\newblock On determinants, matchings, and random algorithms.
\newblock In {\em FCT}, volume~79, pages 565--574, 1979.

\bibitem{M13}
Aleksander Madry.
\newblock Navigating central path with electrical flows: From flows to
  matchings, and back.
\newblock In {\em 54th Annual {IEEE} Symposium on Foundations of Computer
  Science, {FOCS} 2013, 26-29 October, 2013, Berkeley, CA, {USA}}, pages
  253--262, 2013.

\bibitem{MS04}
Marcin Mucha and Piotr Sankowski.
\newblock Maximum matchings via gaussian elimination.
\newblock In {\em 45th Annual IEEE Symposium on Foundations of Computer
  Science}, pages 248--255. IEEE, 2004.

\bibitem{MVV87}
Ketan Mulmuley, Umesh~V Vazirani, and Vijay~V Vazirani.
\newblock Matching is as easy as matrix inversion.
\newblock In {\em Proceedings of the nineteenth annual ACM symposium on Theory
  of computing}, pages 345--354. ACM, 1987.

\bibitem{N09}
Noam Nisan.
\newblock Auction algorithm for bipartite matching.
\newblock In {\em Turing's Invisible Hand blog,
  https://agtb.wordpress.com/2009/07/13/auction-algorithm-for-bipartite-matching/},
  2009.

\end{thebibliography}

\end{document}